\documentclass[10pt,journal,compsoc]{IEEEtran} 

\usepackage{amsmath}
\usepackage{amssymb}
\usepackage{amsthm}
\usepackage{times}
\usepackage[dvipdfmx]{graphicx}
\usepackage{epstopdf}
\usepackage{floatrow}
\usepackage{float}
\usepackage[caption=false]{subfig}
\usepackage{array}
\usepackage{multirow}
\usepackage{algorithm}
\usepackage[noend]{algpseudocode}
\usepackage{hyperref}
\usepackage{hypcap}
\usepackage{afterpage}
\usepackage[justification = justified]{caption}
\usepackage{url}
\makeatletter
\g@addto@macro{\UrlBreaks}{\UrlOrds}
\makeatother

\ifCLASSOPTIONcompsoc
  \usepackage[noadjust]{cite}
\else
  \usepackage[noadjust]{cite}
\fi

\newtheorem{prop}{Proposition}
\newbox\tablebox    \newdimen\tablewidth
\def\leaderfil{\leaders\hbox to 100pt{\hss.\hss}\hfil}

\def\tablenote#1 #2\par{\begingroup \parindent=0.8em
    \abovedisplayshortskip=0pt\belowdisplayshortskip=0pt
    \noindent
    $$\hss\vbox{\hsize\tablewidth \hangindent=\parindent \hangafter=1 \noindent
    \hbox to \parindent{\sup{\rm #1}\hss}\strut#2\strut\par}\hss$$
    \endgroup}
\def\doubleline{\vskip 3pt\hrule \vskip 1.5pt \hrule \vskip 5pt}

\begin{document}
\title{Combinatorial Optimization by Decomposition on Hybrid CPU--non-CPU Solver Architectures}
%%Problem Decomposition for NP-Complete problems in a CPU-non-CPU Hybrid  Architecture

\author{
    \IEEEauthorblockN{Ali Narimani\IEEEauthorrefmark{1}, Seyed Saeed Changiz Rezaei\IEEEauthorrefmark{2}, Arman Zaribafiyan\IEEEauthorrefmark{3}} \\
    \IEEEauthorblockA{1QB Information Technologies (1QBit), 458-550 Burrard Street, Vancouver, British Columbia, V6C 2B5, Canada \\
    Email:  \IEEEauthorrefmark{1}   ali.narimani@1qbit.com, 
    \IEEEauthorrefmark{2}saeed.rezaei@1qbit.com, 
    \IEEEauthorrefmark{3}arman.zaribafiyan@1qbit.com}
}
              
%\thanks{All authors are with 1QB Information Technologies (1QBit), Vancouver, British Columbia, Canada}
%\affiliation{\emph{1QB Information Technologies (1QBit), 458-550 Burrard Street, Vancouver, British Columbia, V6C 2B5, Canada}}
%\author{Seyed Saeed Changiz Rezaei}
%\affiliation{\emph{1QB Information Technologies (1QBit), 458-550 Burrard Street, Vancouver, British Columbia, V6C 2B5, Canada}}
%\author{Arman Zaribafiyan}
%\affiliation{\emph{1QB Information Technologies (1QBit), 458-550 Burrard Street, Vancouver, British Columbia, V6C 2B5, Canada}}
%

\date{\today}
\IEEEtitleabstractindextext{%
\begin{abstract}

The advent of new special-purpose hardware such as FPGA or ASIC-based annealers and quantum processors has shown potential in solving certain families of complex combinatorial optimization problems more efficiently than conventional CPUs. We show that to address an industrial optimization problem, a hybrid architecture of CPUs and non-CPU devices is inevitable. In this paper, we propose problem decomposition as an effective method for designing a hybrid CPU--non-CPU optimization solver. We introduce the required algorithmic elements for making problem decomposition a viable approach in meeting the real-world constraints such as communication time and the potential higher cost of using non-CPU hardware. We then turn to the well-known maximum clique problem, and propose a new method of decomposition for this problem. Our method enables us to solve the maximum clique problem on very large graphs using non-CPU hardware that is considerably smaller than the size of the graph. As an example, we show that the maximum clique problem on the com-Amazon graph, with 334,863 vertices and 925,872 edges, can be solved with a single call to a device that can embed a fully connected graph of size at least 21 nodes, such as the D-Wave 2000Q. We also show that our proposed problem decomposition approach can improve the runtime of two of the best-known classical algorithms for large, sparse graphs, namely PMC and BBMCSP, by orders of magnitude. In the light of our study, we believe that new non-CPU hardware that is small in size could become competitive with CPUs if it could be either mass produced and highly parallelized, or able to provide high-quality solutions to specific, small-sized problems significantly faster than CPUs.

\end{abstract}

% Note that keywords are not normally used for peer review papers.
}

\maketitle

\section{Introduction}

Discrete optimization problems lie at the heart of many studies in operations research and computer science (\cite{blazewicz2013scheduling,kouvelis2013robust}), as well as 
a diverse range of problems in various industries. Crew scheduling problem \cite{kasirzadeh2017airline},  vehicle routing \cite{UPS}, 
anomaly detection~\cite{NASA}, optimal trading trajectory \cite{usOTT}, job shop scheduling \cite{usJSP}, 
prime number factorization \cite{usPrimeFac}, molecular similarity \cite{usGS}, and the kidney exchange problem \cite{kidney} are all examples of discrete optimization problems encountered in real-world applications. Finding an optimum or near-optimum solution for these 
problems leads not only to more efficient outcomes, but also to saving lives, building greener industries, and developing procedures that can lead to increased work satisfaction. 

In spite of the diverse applications and profound impact the solutions to these problems can have, a large class of these problems remain intractable for conventional computers. This 
intractability stems from the large space of possible solutions, and the high computational cost for reducing this space \cite{NoC}. These 
characteristics have led to extensive research on the design and development of both exact and heuristic algorithms that exploit the structure of the
specific problem at hand to either solve these 
problems to optimality, or find high-quality solutions in a reasonable amount of time (e.g., see \cite{san2016new}, \cite{reviewjava}, and \cite{lewis2015guide}).

Alongside research in algorithm design and optimized software, building quantum computers that work based on a new paradigm of computation, such as
D-Wave Systems' quantum annealer \cite{Dwave}, or specialized classical hardware for optimization problems, such as Fujitsu's
digital annealer \cite{fujitsu}, has been a highly active field of research in recent years. 
All of the problems described above can potentially be solved with these devices after the problem has been transformed into a quadratic unconstrained binary optimization (QUBO) problem (see Ref.~\cite{ising}), and these quantum and digital annealers serve as good examples of what we refer to as ``non-CPU" hardware in this paper. 

The arrival of new, specialized hardware calls for new approaches to solving  optimization problems, many of which simultaneously harness the power 
of conventional CPUs and emerging new technologies. 
In one such approach, CPUs are used for pre- and post-processing steps, while solving the problem is  left entirely to the
non-CPU device. The CPUs then handle tasks such as converting the problems into an acceptable format, or 
analyzing the results received from the non-CPU device, without taking an active part in solving the problem.

In this paper, we focus on a different approach that is based on problem decomposition. In this approach, the original problem is decomposed into smaller-sized problems, extending the  
scope of the hardware to larger-sized problems. However, the practical use of problem decomposition depends on a multitude of 
factors. We lay out the foundations of using problem decomposition
in a hybrid CPU/non-CPU architecture in Sec.~\ref{Sec:ProbDecom}, and explain some critical characteristics that are essential for a practical problem
decomposition method within such an architecture. We then focus on a specific NP-hard problem, namely the maximum clique problem, provide and explain the formal
definition of the problem in Sec.~\ref{Sec:MaxCliqueDef}, and propose a new 
problem decomposition method for this problem in Sec.~\ref{Sec:ProbDecomMaxClique}. Sec.~\ref{Sec:ResDis} showcases the potential of our 
approach in extending the applicability of new devices to large and challenging problems, and Sec.~\ref{Sec:Dis} summarizes our results 
and presents directions for future study.

%%====================================================================================
%%====================================================================================
\section{Using a Hybrid Architecture for Hard Optimization Problems}
\label{Sec:ProbDecom}

\begin{figure*}[h]
	\centering
		\includegraphics[width=\textwidth]{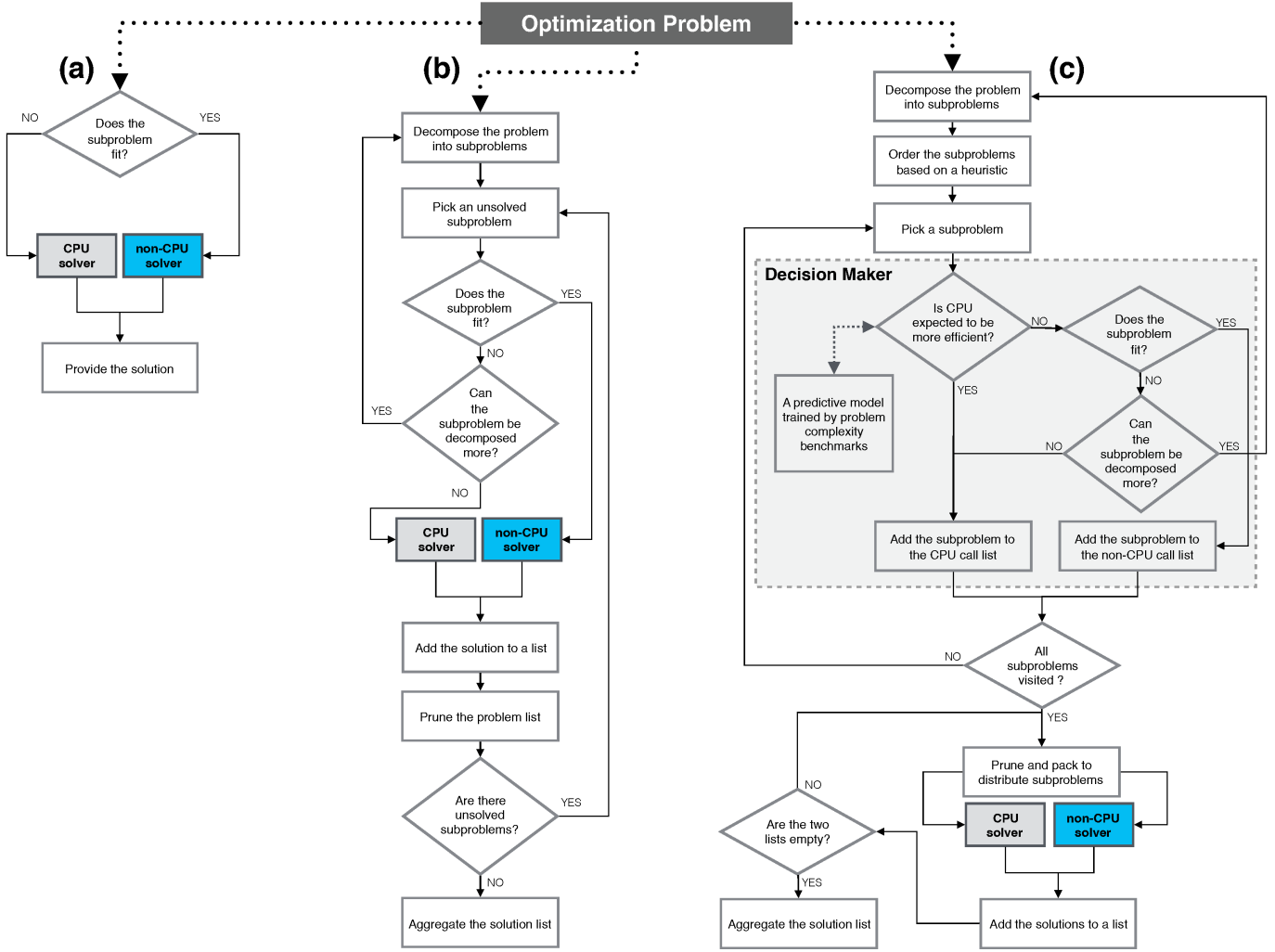}
	\caption{Three different approaches towards a hybrid CPU--non-CPU architecture for solving combinatorial optimization problems. 
	Each approach is illustrated in a separate flowchart.}
	\label{fig:flowcharts}
\end{figure*}

As new hardware is designed and built for solving optimization problems, one key question is how to optimally distribute the tasks between a conventional CPU and this new hardware (see, e.g.,~\cite{NASAdecomp}, 
\cite{dWaveDecomp}). These new hardware devices are designed and tuned to address a specific problem efficiently. However, 
the process of solving an optimization problem involves some pre- and post-processing that might not be possible 
to perform on the application-specific non-CPU hardware. The pre-processing steps include the process of reading the input problem, which is quite 
likely available in a format that is most easily read by classical CPUs, as well as embedding that problem 
into the hardware architecture of the non-CPU device. Therefore, the use of a hybrid architecture that combines CPU and non-CPU 
resources is inevitable. The simplest hybrid methods also use a low-complexity, classical,  local 
search algorithm to further optimize the results of the non-CPU device as a post-processing step. 

This simple picture was used in the early days of non-CPU solver development. However, not all problems are well-suited for a non-CPU device. Furthermore, when large optimization problems are decomposed into smaller subproblems, 
each of the subproblems might exhibit different complexity characteristics. This means that in any given problem, there might be subproblems that 
are better handled by CPU-based algorithms. This argument, together with the fact that usually a single call to a non-CPU device
will cost more than using a CPU, emphasizes the importance of identifying the best use of each device for each problem.
Thus, the CPU should also be responsible for identifying which pieces of the problem are 
best suited for which solver. 

Fig. \ref{fig:flowcharts} illustrates three different hybrid approaches to using a CPU-based and a non-CPU-based solver to solve
an optimization problem. Flowchart (a) represents the simplest hybrid method, which has the lowest level of sophistication in 
distributing tasks between the two hardware devices. It solves the problem at hand using the non-CPU solver only if the size of the problem is less than or equal to the size of the solver.
In this approach, all problems are meant to be solved 
using the non-CPU device unless they do not fit on the hardware for some technical reason. The CPU's function is to carry the 
pre- and post-processing tasks as well as to solve the problems that do not fit on the non-CPU hardware.
Flowchart (b) adds a level of sophistication in that it involves decomposing every subproblem until it either fits the non-CPU hardware, or proves to be difficult to decompose further, in which case it uses a CPU to solve the problem. 
Finally, the method we propose is depicted in (c). It is a hybrid system that uses the idea of problem decomposition in (b), but augments it with
a decision maker and a method that assigns optimization bounds to each subproblem. These additional steps are necessary for the 
practical use of decomposition techniques in hybrid architectures.

However, as we will demonstrate, not every method of decomposition will be beneficial in a hybrid CPU/non-CPU architecture. 
For this method to work best in such a scenario, we propose the following  
requirements:
\begin{itemize}
  \item the number of generated subproblems should remain a polynomial function of the input;
  \item the CPU time for finding subproblems should scale polynomially with the input size.
\end{itemize}

Given these two conditions, the total time spent on solving a problem will remain tractable if the new hardware is capable of efficiently solving problems of a specific type. More precisely, the total computation time in a hybrid architecture can be broken into three components:
\begin{equation}
T_{\text {total}} = t_{\text {CPU}} + t_{\text {comm}} + t_{\text {non-CPU}}.
\end{equation}
Here, $t_{\text {CPU}}$ is the total time spent using the CPU. It consists of decomposing the original problem, solving a fraction of the subproblems that are
not well-suited for the non-CPU hardware, and converting the remaining subproblems into 
an acceptable input format for the new hardware (e.g.,~a QUBO formulation for a device like the D-Wave 2000Q or Fujitsu's digital annealer).  The amount of time devoted to the communication between a CPU and the new hardware is denoted by $t_{\text {comm}}$. This time is proportional to the number of calls made from the CPU to the hardware (which, in itself, is less than
the total number of subproblems, as we will explain shortly).
Furthermore, $t_{\text {non-CPU}}$ is the total time that it takes for the non-CPU hardware to solve all of the subproblems that it receives. 

Given the two requirements for problem decomposition, $ t_{\text {CPU}} + t_{\text {comm}}$ remains polynomial, and using the hybrid architecture
will be justified if the non-CPU hardware is capable of solving the assigned problems significantly more efficiently than a CPU. 

%%====================================================================

\subsection{Problem Decomposition in a Hybrid Architecture}

Algorithm~\ref{alg:decompose} comprises our proposed procedure for using problem decomposition in a hybrid architecture. 
This algorithm takes a $problem$ of size $N$, the size of the non-CPU hardware $nonCPU\_size$, and the
maximum number of times to apply the decomposition method {\tt decomposition\_level} as input arguments. 
In this pseudocode, {\tt solve\_CPU(.)} and {\tt solve\_nonCPU(.)} denote subroutines that solve problems on classical
and non-CPU hardware, respectively. 

At the beginning, the algorithm checks whether a given $problem$ is ``well-suited'' for the non-CPU hardware. 
We define a ``well-suited" problem for a non-CPU hardware device as a problem that is expected to be solved faster on a 
non-CPU device compared to a CPU. This step is performed by the ``decision maker'' (explained in Sec.~\ref{Sec:DecMaker}). 

The algorithm then proceeds to decompose 
the problem only if the entire $problem$ is not well-suited for the non-CPU hardware. When a problem is sent to the 
{\tt do\_decompose(.)} method, it is broken into smaller-sized subproblems, and each subproblem is tagged with an upper bound, 
in the case of maximization, or a lower bound, in that of minimization. These bounds will be used later to 
reduce the number of calls to the non-CPU hardware. This decomposition step can be performed a single time, 
or iteratively up to $decomposition\_level$ times. 
%If the decomposition method generates a linear number of subproblems for each input, then the total number of subproblems 
%at each decomposition level $\ell$ will be $\mathcal O(N^{\ell})$. 

After the original problem is decomposed, every new problem in the 
$subProblem$ list is checked by the decision maker. The well-suited problems are stored in $nonCPU\_subproblems$, and the 
rest are placed in the $CPU\_subproblems$ list.
After the full decomposition has bee achieved, the problems in $nonCPU\_subproblems$ are sent to 
the {\tt PruneAndPack(.)} subroutine. This subroutine ignores the problems with an upper bound 
(lower bound) less than (greater than) the best found solution by the {\tt solve\_CPU(.)} and {\tt solve\_nonCPU(.)} methods, and continues to pack in the
rest of the subproblems until the size of the non-CPU hardware has been maxed out. These are necessary steps for minimizing
the number of calls to the non-CPU hardware, and thus minimizing the communication time. 
At the final step, the results of all of the solved subproblems are combined and analyzed using the {\tt Aggregate(.)} subroutine.  

\begin{algorithm*}
	
	\caption{\small \, Problem Decomposition Algorithm}	\label{alg:decompose}
	\begin{algorithmic}[1]
	\Procedure{\tt Decompose}{$problem$, $nonCPU\_size$, $decomposition\_level$}
	\If{$problem$ is not well-suited }
	\State{\textbf{return} {\tt Solve\_CPU}($problem$)}
	\EndIf
	\If{$problem\_size < nonCPU\_size$}
	\State{\textbf{return} {\tt Solve\_nonCPU}($problem$)}
	\EndIf
	\State{$subproblems\leftarrow$ {\tt Do\_decompose}($problem$)}
	\State{$CPU\_subproblems\leftarrow \emptyset$}
	\State{$nonCPU\_subproblems\leftarrow \emptyset$}
	\For{$problem_i$ in $subproblems$}
	\If{$problem_i$ is not well-suited}
	\State{$CPU\_subproblems \leftarrow CPU\_subproblems\cup problem_i$}
	\Else
	\If{$problem\_size > nonCPU\_size$}
	\State{$level\leftarrow 0$}
	\State{$buffer\_list\leftarrow problem$}
	\While{$buffer\_list\_size\neq 0$ and $level<decomposition\_level$}
	\For{$p$ in $buffer\_list$}
	\State{$new\_buffer\_list\leftarrow$ {\tt Do\_decompose}($p$)}
	\EndFor
	\State{$buffer\_list \leftarrow\emptyset$}
	\For{$p$ in $new\_buffer\_list$}
	\If{$p$ is not well-suited}
	\State{$CPU\_subproblems\leftarrow CPU\_subproblems\cup p$}
	\Else
	\If{$p\_size > nonCPU\_size$}
	\State{$buffer\_list\leftarrow buffer\_list\cup p$}
	\Else
	\State{$nonCPU\_subproblems\leftarrow nonCPU\_subproblems \cup p$}
	\EndIf
	\EndIf
	\EndFor
	\State{$level \leftarrow level + 1$}
	\EndWhile
	\Else
	\State{$nonCPU\_subproblems\leftarrow nonCPU\_subproblems\cup problem_i$}
	\EndIf
	\EndIf
	\EndFor
	\For{$p$ in $CPU\_subproblems$}
	\State{$CPU\_result \leftarrow$ {\tt Solve\_CPU}($p$)}
	\EndFor
	\While{$nonCPU\_problems$ is not empty}
	\State{$packed\_problems \leftarrow$ {\tt PruneAndPack}($nonCPU\_problems$)}\label{line:pack}
	\State{$nonCPU\_result \leftarrow$ {\tt Solve\_nonCPU}($packed\_problems$)}
	\EndWhile
	\State{\textbf{return} {\tt Aggregate}($nonCPU\_result, cpu\_result$)}
	\EndProcedure
\end{algorithmic}
\end{algorithm*}

\subsection{Decision Maker} \label{Sec:DecMaker}
There is always an overhead cost in converting each subproblem into an acceptable format for the non-CPU hardware, sending the 
correctly formatted subproblems to this hardware, and finally receiving the answers. It is hence logical to send the subproblems to 
a ``decision maker'' before preparing them for the new hardware. In an ideal scenario, this decision maker will have access to a portfolio of 
classical algorithms, along with the specifications of the non-CPU hardware. Based on this information, the decision maker will be able to decide whether
a given problem is well-suited for the non-CPU hardware. These decisions may be achieved via either some simple characteristics of the problem, or 
through intelligent machine-learning models with good predictive power, depending on  the case at hand.

\section{Specific Case Study: Maximum Clique} \label{Sec:MaxCliqueDef}

Now that we have laid out the specifics of our proposal for problem decomposition in a hybrid architecture, we apply
this method to the maximum clique problem. We begin by explaining the graph theory notation and necessary definitions, along with
a few real-world applications for the maximum clique problem. 

A graph \mbox{$G = (V, E)$} consists of a finite set \mbox{$V = \{v_1, v_2, \ldots, v_n\}$} of vertices and a set $E\subseteq V\times V$ of edges. Two 
distinct vertices $v_i$ and $v_j$ are adjacent if $\{v_i, v_j\}\in E$. The \emph{neighbourhood} of a vertex $v$ is denoted by $\mathcal N(v)$, 
and is the subset of vertices of $G$ which are adjacent to $v$.  The degree of a vertex $v$ is the cardinality of $\mathcal N(v)$, and is 
denoted by $d(v)$. The maximum degree and minimum degree of a graph are denoted by $\Delta(G)$ and $\delta(G)$, respectively.

The \emph{subgraph} of $G$ \emph{induced} by a subset of vertices $U\subseteq V$ is denoted by $G[U]$ and consists of the vertex set $U$, 
and the edge set defined by 
\[
E(G[U]) = \{\{u_i, u_j\}~|~ u_i,~u_j \in U,~\{u_i, u_j\}\in E(G)\}.
\]

A \emph{complete subgraph}, or a \emph{clique}, of $G$ is a subgraph of $G$ where every pair of its vertices are adjacent. The size of a 
maximum clique in a graph $G$ is called the \emph{clique number of} $G$ and is denoted by $\omega(G)$. An independent set of $G$, on the 
other hand, is a set of pairwise nonadjacent vertices. As every clique of a graph is an independent set of the complement graph, 
one can find a maximum independent set of a graph by simply solving the maximum clique problem in its complement.

A node-weighted graph $G$ is a graph that is augmented with a set of positive weights $W = \{w_1, w_2, \ldots, w_n\}$ assigned to each node. The 
maximum weighted-clique problem is the task of finding a clique with the largest sum of weights on its nodes.

%The maximum clique problem is also related to another NP-hard problem on graphs, named graph coloring. A proper vertex coloring aims at assigning colors to vertices such that all pairwise adjacent vertices receive different colors. More precisely, an assignment
%\[
%c(.): V\rightarrow\mathbb{N}
%\]
%is a proper coloring if for every pair of adjacent vertices $v_i$ and $v_j$, we have $c(v_i) \neq c(v_j)$. The goal of a \emph{vertex coloring problem} is to find an assignment with the minimum number of colors. The minimum number of colors obtained from a vertex coloring problem is called the \emph{chromatic number} of $G$ and is denoted by $\chi(G)$. 
% The chromatic number is an upper bound on the size of the maximum clique of a graph, $\omega(G) \leq \chi(G)$, since each of the vertices of a clique receive a unique color in a proper coloring.

Many real-world applications have been proposed in the literature for the maximum clique and maximum independent set problems. 
One commonly suggested application is community detection for social network analysis \cite{CommDet}. Even though cliques
are known to be too restrictive for finding communities in a network, they prove to be useful in finding overlapping communities. 
Another example is the finding of the largest set of correlated/uncorrelated instruments in financial markets.
This problem can be readily modelled as a maximum clique problem, and it plays an important role in risk management 
and the design of diversified portfolios (see \cite{finance} and \cite{marketGraph}). 
Recent studies have shown some merit in using a weighted maximum-clique finder for drug discovery purposes 
(see \cite{bio} and \cite{graphSimilarity}). In these studies, the structures of molecules are stored as graphs, and the properties of 
unknown molecules are predicted by solving the maximum common subgraph problem using the graph representations of the molecules. 
Aside from the proposed industrial applications, the clique problem 
is one of the better-studied NP-hard problems, and there exist powerful heuristic and exact algorithms for solving the maximum clique 
problem in the literature (see, e.g.,~\cite{reviewjava}, \cite{heuristic}, and \cite{para}). It is, therefore, beneficial
to map a part of, or an entire, optimization problem into a clique problem and benefit from the runtime of these algorithms (see, e.g.,~\cite{mapApp}).

\section{Problem Decomposition for the Maximum Clique Problem} \label{Sec:ProbDecomMaxClique}
In this section, we explain the details of two problem decomposition methods for the maximum clique problem. 
The first approach is based on the branch-and-bound framework and is similar to what is dubbed ``vertex splitting'' in 
Ref.~\cite{newAnnealerClique}.  This method is briefly explained in Sec.~\ref{Sec:BnB}, followed by a discussion on
why it fails to meet the problem decomposition requirements of Sec.~\ref{Sec:ProbDecom}. We then present our own method in Sec.~\ref{Sec:k-core} and prove that it is an effective problem  decomposition method, 
that is,~it generates a polynomial number of subproblems and requires polynomial computational complexity to generate
each subproblem.

\subsection{Branch and Bound}\label{Sec:BnB}
The branch-and-bound technique (BnB) is a commonly used method in exact algorithms for solving the maximum clique problem
(see Ref.~\cite{reviewjava} for a comprehensive review on the subject).  At a very high level, BnB consists of three main procedures 
that are repeatedly applied to a subgraph of the entire graph until the size of the maximum clique is found.  
The main procedures of a BnB approach consist of: (a) ordering the vertices in a given subproblem and adding the highest-priority vertex to the
solution list; (b) finding the space of feasible solutions based on the vertices in the solution list; and (c) assigning upper bounds to each subproblem. 
\begin{figure*}
	\centering
		\includegraphics[scale=0.88]{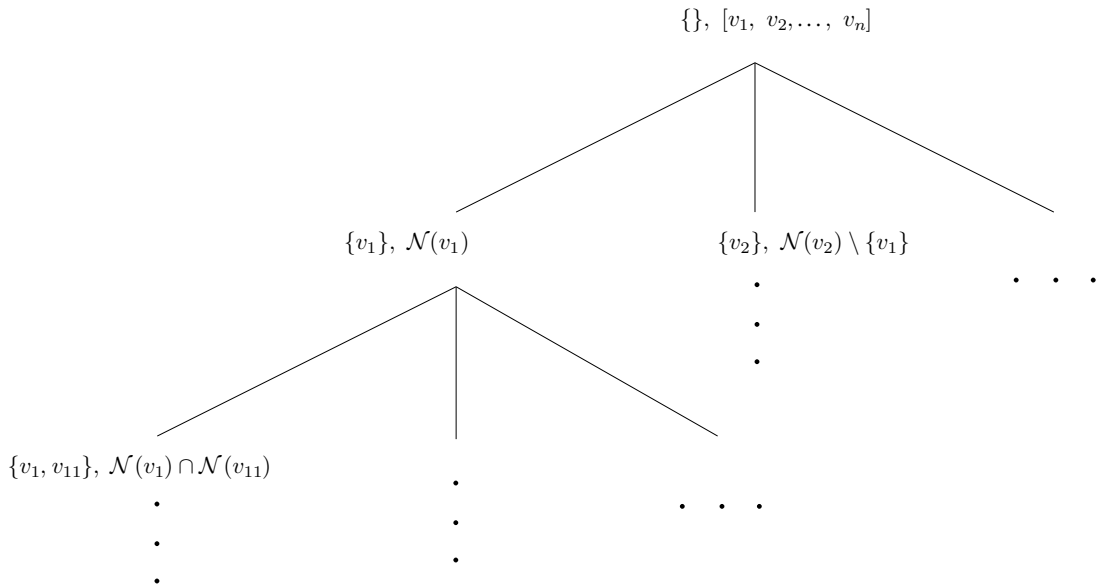}
	\caption{Abstract illustration of a branch-and-bound search tree}
	\label{fig:bnb}
\end{figure*}
Fig. \ref{fig:bnb} shows a schematic representation of the steps involved in traversing the BnB search tree. In the first step, all of the vertices of the graph are listed
at the root of the tree, representing the space of feasible solutions, along with an empty set that will contain 
possible solutions as the algorithm traverses the search tree (we will call this set ``\textit{growing-clique}'').
The vertices inside the feasible space are ordered based on some criteria (e.g.,~increasing/decreasing degree, or the sum of the degree of the 
neighbours of a vertex \cite{TomitaMCR}), 
and the highest-priority vertex ($v_{\rm 1}$) is chosen as the ``branching node". The branching node is added
to \textit{growing-clique} and the neighbourhood of this node ($\mathcal N(v_{\rm 1})$) is chosen as the new space of feasible solutions. 
This procedure continues until the domain of feasible solutions becomes an empty list, indicating that \mbox{\textit{growing-clique}} now contains a
maximal clique. If the size of this clique is larger than the best existing solution, the best solution is updated. The number of nodes in the 
BnB tree is greatly reduced by applying some upper bounds based on graph colouring \cite{TomitaSeki} or Max-SAT reasoning \cite{maxSAT}. These
upper bounds prune the tree if the upper bound on the size of the clique inside a feasible space is smaller than the best found solution (minus the size of  
\textit{growing-clique}).

In a fully classical approach, the entire BnB search tree is explored via a classical computer. On the other hand, 
some of the work can be offloaded to the non-conventional hardware in the hybrid scenario.
More precisely, one can stop traversing a particular branch of the search tree when the size of the subproblem under consideration  
becomes smaller than the capacity of the non-conventional hardware (see, e.g.,~\cite{newAnnealerClique}). Although this idea 
can combine the two hardware devices in an elegant and coherent way, it suffers from two main drawbacks. It creates an exponential number
of subproblems (see Fig.~$6$ in Ref.~\cite{newAnnealerClique}), and, in the worst case, it takes an exponential amount of time to traverse the search 
tree until the size of the subproblem becomes smaller than the capacity of the non-conventional hardware. 

\subsection{A Proposed Method for Problem Decomposition}\label{Sec:k-core}

\begin{algorithm}
	\caption{\small \, $k$-Core Computation and Degeneracy Ordering}	\label{alg:kCorePatent}
	\begin{algorithmic}[1]
\Procedure{\tt K-Core\_Compute}{undirected weighted graph $G = (V, E)$}
		\State $sorted\_nodes$ $\leftarrow$ an empty list
		\For{ every $v \in G$}
		\State{$K(v) \leftarrow d(v)$}
		\EndFor
		\State $V'\leftarrow$ order vertex set $V$ by non-decreasing vertex degree
		\While{$V'$ is not empty}
		\State select a vertex $v$ with minimum $k$-core, and find its neighbourhood $\mathcal N(v)$
		\For{ every $u\in\mathcal N(v)$}
		\If{$K(u) > K(v)$}
		\State{$K(u) \leftarrow K(u) - 1$}
		\EndIf
		\EndFor
		\State{$sorted\_nodes \leftarrow  sorted\_nodes \cup  \{v\}$}
		\State{$V'\leftarrow V'\setminus \{v\}$}
		\EndWhile
		
		\State \textbf{return} $sorted\_nodes$, $k$-$core$
		\EndProcedure
\end{algorithmic}
\end{algorithm}

In this section, we explain our proposed problem decomposition method, which is much more effective than BnB (explained in the previous section). 
We show, in particular, that our proposed method generates a much smaller number of subproblems
compared to BnB, and that these subproblems can be obtained via an efficient $\mathcal O(E)$ algorithm.

Our method begins by sorting the vertices of the graph based on their $k$-core number. Ref.~\cite{kCore} details the formal $k$-core definition,
and proposes an efficient $\mathcal O(E)$ algorithm for calculating the $k$-core number of the vertices of a graph. Intuitively, the $k$-core number of 
a vertex $v$ is equal to $k$ if it has at least $k$ neighbours of a degree higher than or equal to $k$, \textit{and} not more than $k$ neighbours of a degree
higher than or equal to $k+1$. We denote the core number of a vertex $v$ by $K(v)$. 

The core number of a graph $G$, denoted by $K(G)$, is the highest-order core of its vertices. $K(G)$ is always upper bounded by the maximum degree
of the vertices of the graph $\Delta(G)$, and the minimum core number of the vertices is always equal to the minimum degree $\delta(G)$.
A \emph{degeneracy}, or \emph{k-core ordering}, of the vertices of a graph $G$ is a non-decreasing ordering of the vertices of $G$ based on their core numbers. Algorithm~\ref{alg:kCorePatent} is a method for finding the $k$-core ordering of the vertices along with their $k$-core numbers. 
The following proposition shows that, given a degeneracy ordering for the vertices of the graph, one can decompose the maximum clique problem into a linear number of subproblems. Our proposed method is based on this proposition. \newline
\begin{prop}\label{prop:oracleCounts}
For a graph $G$ of size $n$, one can decompose the maximum clique problem in $G$ into at most $n - K(G) + 1$ subproblems, each of which is upper-bounded in size by $K(G)$.
\end{prop}
\begin{proof}
Let $d(v)$ be the degree of vertex $v$ in $G$. From Algorithm \ref{alg:kCorePatent} (lines $8$--$11$), the number of vertices that are adjacent to $v$ and precede vertex $v$ in $sorted\_nodes$ is greater than or equal to $d(v) - K(v)$. Therefore, the number of vertices that appear after $v$ in this ordering is 
upper-bounded by $K(v)$. Using this fact, the algorithm starts from the last $K(G)$  vertices of $sorted\_node$, and solves the maximum clique on that 
induced subgraph. It then moves towards the beginning of $sorted\_node$ vertex by vertex. Each time it takes a root vertex $w$ and forms a new subproblem 
by finding the adjacent vertices that are listed after $w$ in  $sorted\_node$. The size of these subproblems is upper-bounded by $K(w)$, which itself is 
upper-bounded by $K(G)$, and the number of the subproblems created in this way is exactly $n - K(G) + 1$. 

Since we have 
\begin{align*}
\omega(G) \leq K(G) + 1 \leq \Delta(G) + 1,
\end{align*}
one can  stop the procedure as soon as the size of the clique becomes larger than or equal to the $k$-core number of a root vertex.
\end{proof}

To illustrate, consider the 6-cycle with a chord shown in Fig.~\ref{fig:C6}. In the first step, the vertices are 
ordered based on their core numbers, according to Algorithm \ref{alg:kCorePatent}: $sorted\_nodes$ = [c, b, e, f, d, a].

\begin{figure}[H]
	\centering
		\includegraphics[scale=0.7]{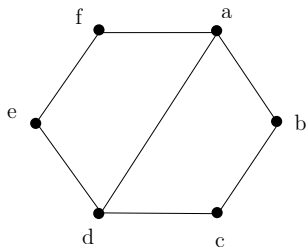}
	\caption{A 6-cycle graph with a chord}
	\label{fig:C6}
\end{figure}

Following our proposed algorithm,  we first consider the subgraph induced by the last $K(G) = 2$ vertices from the list, that is, $[a,d]$. 
Solving the maximum clique problem in this subgraph results in a lower bound on the size of the maximum clique, that is,~$\omega(G) \geq 2$.
After this step, we proceed by considering the vertices one by one from the end of  $sorted\_nodes$ to form the subproblems that follow:
\begin{align*}
root\_node \quad     &  subproblem \\
 f   \qquad  \quad \,  &   \quad  \quad    \{\}  \\
 e  \qquad  \quad \,  &   \quad \,\,    \{f,d\} \\ 
 b \qquad   \quad \,  &   \quad  \,\,\,\,\,    \{a\}  \\
 c \qquad   \quad \,  &   \quad \,\,    \{b,d\}    \,.
\end{align*}
Among these subproblems, only the two with \mbox{$root\_node = \{e,c\}$} need to be examined by the clique solver, since the size of the 
other subproblems is less than or equal to the size of the largest clique found. This example shows how a problem of size six can
be broken down to three problems of size two.

As a final note, the $k$-core decomposition takes $\mathcal O(E)$ time, and constructing the resulting subproblems takes
$\mathcal O(N^2)$. The entire process takes \mbox{$\mathcal O(E + N^2)$} time at the first level, and $\mathcal O(N^{(\ell +1)})$
time at \mbox{$decomposition\_level = \ell$}. The maximum number of subproblems can grow up to $N^\ell$ at \mbox{$decomposition\_level = \ell$}.

\section{Results} \label{Sec:ResDis}
In this section, we discuss our numerical results for different scenarios in terms of density and the size of the underlying graph. 
In particular, we study the effect of the graph core number, $K(G)$, and the density of the graph on the number of generated 
subproblems. In the fully classical approach, we also compare the running time of 
our proposed algorithm with  state-of-the-art methods for solving the maximum clique problem in the large, sparse graphs.  
It is worth noting that $k$-core decomposition 
is widely used in exact maximum clique solvers as a means to find computationally inexpensive and relatively tight upper bounds in large, sparse graphs. 
However, to the best of our knowledge, no one has used $k$-core decomposition as a method of problem decomposition as is proposed 
in this paper (e.g.,~Ref.~\cite{newAnnealerClique} uses it for pruning purposes and BnB for decomposition).  

\subsection{Large and Sparse Graphs}
The importance of Proposition~\ref{prop:oracleCounts} is more pronounced when we consider the standard large, sparse graphs 
listed in Table~\ref{table:results}. For each of these graphs, we first perform one round of $k$-core decomposition, and then solve the 
generated subproblems with our own exact maximum clique solver. It is worth mentioning that, after decomposing the original 
problem into sufficiently smaller subproblems, our approach for finding the maximum clique of the smaller subproblems is similar 
to what has been proposed in Ref.~\cite{pmc}.

%%===================================================
%% LARGE SPARSE RESULTS
%%===================================================
\begin{table*}
\centering
\scriptsize
\caption{The number of subproblems and the runtime results (in seconds) for exact maximum clique solvers (1QBit's solver, PMC, and BBMCSP) on network graphs from the \href{https://snap.stanford.edu/data/}{Stanford Large Network Dataset Collection} \cite{Stanford} and \href{http://networkrepository.com/}{Network Repository} \cite{netrepo}. ``MaxClique" refers to the size of the maximum clique. All three codes are implemented in C++ and run on an Ubuntu machine with 32 GB of memory and an Intel 3.60 GHz processor.}
\label{table:results}
\begin{tabular}{ l r r r r | r r r | r}
\noalign{\doubleline}
%\noalign{\vskip 3pt\hrule\vskip 3pt} 
 & & & & & &Runtime (s) & &\cr
Graph Name & Num. of Vertices & Num. of Edges & $K(G)$ & MaxClique & 1QBit Solver & PMC & BBMCSP & Num. of Subprobs. \cr
\noalign{\vskip 3pt\hrule\vskip 3pt} 
\textbf{Stanford Large Network} & & & & & & & &\cr
\textbf{Dataset:} &  & & & & & & &\cr
\noalign{\vskip 3pt\hrule\vskip 4pt} 
ego-Facebook & $4,039$         & $88,234$ & $115$  & $69$ & $\textbf{0.009}$ & $0.04$ & $0.03$ & $367$\cr
ca-CondMat & $23,133$ & $93,468$  & $25$  & $26$ & $\textbf{0.004}$ & $0.03$& $0.02$ & $3$\cr
email-Enron & $36,692$  &   $25,985$& $43$ & $20$ & $\textbf{0.01}$ & $0.3$ & $0.06$ & $2235$\cr
com-Amazon & $334,863$       & $925,872$& $6$ & $7$ & $\textbf{0.06}$ & $0.06$ & $0.2$ & $3$\cr
roadNet-PA & $1,088,092$ & $1,541,898$& $3$ & $4$ & $\textbf{0.1}$ & $2.1$ & $0.4$ & $580$ \cr
com-Youtube  & $1,134,890$       & $2,987,624$& $51$ & $17$ & $\textbf{0.5}$ & $2.0$ & $2.5$ & $24466$ \cr
as-skitter  & $1,696,415$       & $11,095,298$& $111$ & $67$ & $\textbf{0.7}$ & $1.4$ & $5.6$ & $4087$ \cr
roadNet-CA & $1,965,206$       & $2,766,607$& $3$ & $4$ & $\textbf{0.2}$ & $3.9$ & $0.4$ &  $2286$\cr
com-Orkut &                        $3,072,441$ & $117,185,083$& $253$ & $51$ & $\textbf{82}$ & $179$ & $220$ & $741349$\cr
com-LiveJournal & $3,997,962$ & $34,681,189$& $360$ & $327$ & $\textbf{2.1}$ & $2.5$ & $20$ & $25$\cr 
\noalign{\vskip 3pt\hrule\vskip 3pt} 
\textbf{Network Repository Graphs}: & & &  & & & & \cr
\noalign{\vskip 3pt\hrule\vskip 4pt} 
soc-buzznet & $101,163$ & $2,763,066$& $153$ & $31$ & $\textbf{1.9}$ & $14.6$ & $4.0$ & $29484$ \cr
soc-catster & $149,700$ & $5,448,197$& $419$ & $81$ & $\textbf{1.1}$ & $>1 \mathrm{\, h}$ & $5.5$ & $12095$ \cr
delaunay-n20 & $1,048,576$ & $3,145,686$& $4$ & $4$ & $\textbf{1.4}$ & $2.5$ & $2.8$ & $1036595$ \cr
web-wikipedia-growth & $1,870,709$ & $36,532,531$& $206$ & $31$ & $\textbf{18}$ & $397$ & file not supported & $358272$ \cr
delaunay-n21 & $2,097,152$ & $6,291,408$& $4$ & $4$ & $\textbf{2.8}$ & $5.1$ & $\textbf{2.8}$ & $2073021$ \cr
tech-ip & $2,250,498$ & $21,643,497$& $253$ & $4$ & $\textbf{28}$ & $1031$ & $220$ & $222338$ \cr
soc-orkut-dir & $3,072,441$ & $117,185,083$& $253$ & $51$ & $\textbf{74}$ & $188$ & $170$ & $741349$ \cr 
socfb-A-anon & $3,097,165$ & $23,667,394$& $74$ & $25$ & $\textbf{10}$ & $18$ & $30$ & $357836$\cr
soc-livejournal-user-groups & $7,489,073$ & $112,305,407$& $116$ & $9$ & $\textbf{103}$ & $>1 \mathrm{\, h}$ & $1600$ & $2404573$ \cr
aff-orkut-user2groups & $8,730,857$ & $327,036,486$& $471$ & $6$ & $\textbf{852}$ & $>1 \mathrm{\, h}$ & $2400$ & $4173108$ \cr
soc-sinaweibo & $58,655,849$ & $261,321,033$& $193$ & $44$ & $\textbf{93}$ &  $>1 \mathrm{\, h}$ & $1070$ & $713652$ \cr
\noalign{\vskip 3pt\hrule\vskip 4pt} 
\end{tabular} 
\end{table*} 

In the large and sparse regime, the core numbers of the graphs are typically orders of magnitude smaller than the number of vertices
in the graph. 
This implies that non-CPU hardware of a size substantially smaller than the size of the original problem 
can be used to find the maximum clique of these massive graphs.
For example, the \emph{com-Amazon} graph, with 334,863 vertices and 925,872 edges 
has $K(G) = 6$. These facts, combined with Proposition~\ref{prop:oracleCounts}, imply that this graph can be decomposed into $\sim$~334,863
problems of a size $\leq 6$. However, as the table shows, the actual number of subproblems that should be solved is only three, 
since the $k$-core number of the next subproblem drops to a number less than or equal to the size of the largest clique that was found. 
Hence, in this specific case, a single call to a non-CPU hardware device of size greater than $21$ can solve the entire problem,
that is,~this problem can be solved by submitting an effective problem of size 21 to the D-Wave 2000Q chip.

%%===================================================
%% DENSE RESULTS
%%===================================================
Numerical results also indicate that the fully classical runtime of our proposed method is considerably faster than two of the 
best-known algorithms in the literature for large--sparse graphs, namely PMC \cite{pmc} and BBMCSP~\cite{bbmcsp}. Table~\ref{table:results} 
shows that in some instances, our method is orders of magnitude faster than these algorithms.  

\subsection{Hierarchy of minimum degree, maximum degree, max $k$-core, and clique number}
The $k$-core decomposition proves extremely powerful in the large and sparse regime because it dramatically reduces the size of the
problem, and also prunes a good number of the subproblems.
This situation changes as we move towards denser and denser graphs. As  graphs
increase in density, $K(G)$ approaches the graph size, and the size reduction becomes less effective in a single iteration of decomposition. 
It is hence necessary to apply the decomposition method for at least a few iterations in the dense regime.

\begin{figure}
	\centering
		\includegraphics[scale=0.66]{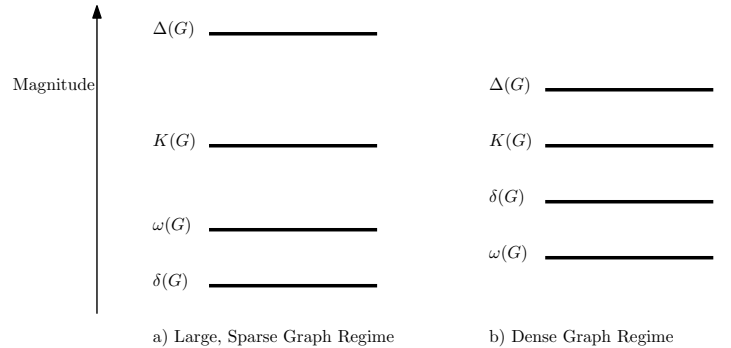}
	\caption{$\omega(G)$, $K(G)$, $\Delta(G)$, and $\delta(G)$ comparisons in different graph density regimes}
	\label{fig:regimes}
\end{figure}

Moreover, the number of subproblems that should be solved also grows as the graphs increase in density. 
This is partially due to there being more rounds of decomposition, and partially to the hierarchy of the clique number $\omega(G)$ and
the maximum and minimum core numbers (shown in Fig.~\ref{fig:regimes}).
In the sparse regime, the clique number $\omega(G)$ lies between the minimum core number $\delta(G)$ 
and the core number of the graph $K(G)$. This  means that all of the subproblems that stem from a root node with a core number 
less than the clique number can be pruned. This phenomenon leads to the effective pruning that is reflected in the number of subproblems
listed in Table \ref{table:results}. On the other hand, 
as shown in Fig.~\ref{fig:regimes}, the minimum core of the graph, i.e,~$\delta(G)$ is larger than the 
clique number in the dense regime. Therefore, core numbers are no longer suitable for upper-bounding purposes, and some other upper-bounding methods
should be used, as we discuss in the next section.

\subsection{Dense Graphs}
As explained in the previous section, the $k$-core number becomes an ineffective upper bound in the case of dense graphs; therefore,
the number of subproblems that should be solved grows as $N^\ell$ for $\ell$ levels of decomposition. Because of this issue, 
and since colouring is an effective upper bound in the dense regime, we used the heuristic DSATUR	 algorithm explained in 
Ref.~\cite{lewis2015guide} to find the colour numbers of each generated subproblem. We then ignored the subproblems with a colour number 
less than the size of the largest clique that has been found from the previous set of subproblems. This technique reduces the number of subproblems by a large factor, 
as can be seen in Table~\ref{table:results2}.
In this table, we present the results for random Erd\H os--R\'enyi graphs of three different sizes and varying densities. For each size and density,
we generated $10$ samples, and decomposed the problems iteratively three levels ($decomposition\_level = 3$). The reported results are
the average of the $10$ samples for each category. ``max", ``min", and ``avg" refer to the maximum, minimum, and average size of the 
generated subproblems at every level. 

\begin{table*}
\centering
\scriptsize
\caption{Problem decomposition statistics for the maximum clique of random 
Erd\H os--R\'enyi graphs of different sizes and densities. 
The first number in parentheses represents size, and the second number is the 
edge presence probability (which ends up being the graph density to a desired level of precision).}
\label{table:results2}
\resizebox{\textwidth}{!}{
\begin{tabular}{ l r r |r r r r |r r r r	|r r r r}
\noalign{\doubleline}
%\noalign{\vskip 3pt\hrule\vskip 3pt} 
& & & &Level 1  & & & & Level 2	& & & & Level 3\cr
Graph Name & $K(G)$ & MaxClique & Num. of Subprobs. & max & min & avg & Num. of Subprobs. & max & min & avg	& Num. of Subprobs. & max & min & avg\cr
\noalign{\vskip 3pt\hrule\vskip 3pt} 
 
ER(200,0.3) & $47.7$   & $7.2$  & $153.3$  & $48$  & $16$ 	& $36$ 	& $74$ 		& $13$ 	&$8$ 	&$10$ 	& 14.1 		&5     	&2            &5\cr
ER(200,0.4) & $66.2$   & $9$    & $134.8$  & $66$  & $27$ 	& $52$	& $460$    	& $26$ 	&$12$      &$19$ 	& 5.4  		&9 	        &6            &8\cr
ER(200,0.5) & $85.3$   & $11$   & $115.7$  & $85$  & $44$ 	& $69$ 	& $2941.2$ 	& $42$ 	&$17$      &$29$ 	& 418.7 		&21    	&12          &16\cr
ER(200,0.6) & $104.2$  & $13.9$ & $96.8$   & $104$ & $62$ 	& $88$ 	& $3856.6$ 	& $60$      &$26$      &$43$ 	& 6685.3 	        &37   	&18          &26\cr
ER(500,0.3) & $127.8$  & $8.5$  & $373.2$  & $128$ & $38$ 	& $93$ 	& $8364.4$ 	& $37$ 	&$12$      &$25$ 	& 217.2 		&10   	&6            &7\cr
ER(500,0.4) & $174.8$  & $10.3$ & $326.2$  & $175$ & $71$ 	& $133$ 	& $23813.6$     & $67$ 	&$18$      &$42$ 	& 11050.9 	&28   	&11          &18\cr
ER(500,0.5) & $223.1$  & $13$   & $277.9$  & $223$ & $113$ 	& $177$ 	& $27437.3$     & $107$ 	&$29$      &$67$ 	& 359735 	         &54  	&17          &33\cr
ER(1000,0.1)& $80.8$   & $5.4$  & $920.2$  & $81$  & $7$ 	& $54$ 	& $985.1$ 	& $7$ 	&$4$ 	&$7$ 	& 658.9 		&5 	        &4            &4\cr
ER(1000,0.2)& $171.9$  & $7.2$  & $829.1$  & $172$ & $33$ 	& $116$ 	& $15666.3$     & $33$ 	&$8$ 	&$22$ 	& 361 		&8 	        &6            &6\cr
ER(1000,0.3)& $265.6$  & $9.1$  & $735.4$  & $266$ & $78$ 	& $188$ 	& $82170.8$     & $76$ 	&$15$      &$44$ 	& 13515.9 	&24    	&8            &15\cr 

\noalign{\vskip 3pt\hrule\vskip 4pt} 
\end{tabular}}
\end{table*}

Notice the significant difference in the $\frac{K(G)}{ \rm{graph \, size}}$ ratio between the sparse graphs
presented in  Table~\ref{table:results} and the relatively dense graphs presented here. Unlike in the sparse regime, the gain in size reduction 
after one level of decomposition is only a factor of few. 
This means that a graph of size $N$ is decomposed into $\sim N$ graphs of smaller but relatively similar size after one level of decomposition. 
This fact hints towards using more levels of decomposition, as with more decomposition, the maximum size of the subproblems decreases.
However, there is usually a tradeoff between the number of generated subproblems and the maximum size of these subproblems, as can be seen
in Table \ref{table:results2}.
It is, therefore, not economical to use problem decomposition for these types of 
graphs in a fully classical approach for solving the maximum clique problem in dense graphs. However, if a specialized non-CPU hardware device becomes significantly faster 
than the performance of CPUs on the original problem, 
this problem decomposition approach will become useful. In fact,
the merit of this approach compared to the BnB-based decomposition, shown in Fig. 6 in Ref.~\cite{newAnnealerClique}, is that it generates 
considerably fewer of subproblems, and that the time for constructing the subproblems is polynomial. 

Fig.~\ref{fig:densityScaling} shows the scaling of the number of subproblems with the density of the graph. For this plot, we assume two
 devices of size $\{45,65\}$, representing an instance of the  D-Wave 2X chip, and the theoretical upper bound on the maximum size of a complete graph embeddable into the new D-Wave 2000Q chip.
The graph size is fixed to 500 in every case, and the points are the average of 10 samples, with error bars showing standard deviation.
For each point, we first run a heuristic on the whole graph, and then prune the subproblems based on their
colour numbers obtained using the DSATUR algorithm of Ref.~\cite{lewis2015guide}. Densities below 0.2 are shaded with a grey band, since the number 
of subproblems for $\{0.05,0.1,0.15\}$ densities is zero. This happens because all of the subproblems
are pruned after the second round of decomposition.
Aside from scaling with respect to density, this plot also shows that a small increase to the size of the non-CPU hardware (e.g.,~from 45 nodes to 65), 
will not have a significant effect on the total number of generated subproblems. In these scenarios, the non-CPU hardware becomes competitive with
classical CPUs only if, in comparison to CPUs, it can either solve a single problem with very high quality and speed, or it can be mass produced and parallelized 
at lower costs.

\begin{figure}[h]
	\centering
		\includegraphics[scale=0.35]{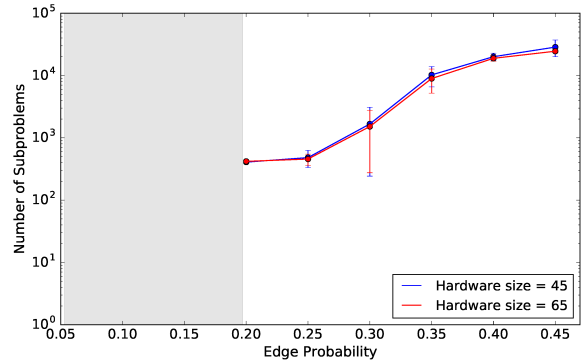}
	\caption{Scaling of the number of subproblems with density. Graph size is fixed to 500 in every case, and every point 
	is the average of 10 samples, with error bars representing sample standard deviation.}
	\label{fig:densityScaling}
\end{figure}

\section{Discussion} \label{Sec:Dis}
We focused on the specific 
case of the maximum clique problem and proposed a method of decomposition for this problem. 
Our proposed decomposition technique is based on the $k$-core decomposition of the input graph. 
The approach is  motivated mostly by the emergence of non-CPU hardware for solving hard problems. This approach
is meant to extend the capabilities of this new hardware for finding the maximum clique of large graphs.
While the size of generated subproblems
is greatly reduced in the case of sparse graphs after a single level of decomposition, an effective size reduction happens only after
multiple levels of decomposition in the dense regime. Compared to the branch-and-bound method, this method generates considerably fewer subproblems, \textit{and} creates these subproblems in polynomial time.

We believe that further research on finding tighter upper bounds on the size of the maximum clique in 
each subproblem would be extremely useful. Tighter upper bounds make it possible to attain more levels of decomposition, and hence 
reduce the problem size, without generating too many subproblems. 

In the fully classical approach, there is a chance that combining integer programming solvers for the maximum clique problem with our
proposed method can lead to better runtimes for dense graphs, or for large, sparse graphs with highly dense $k$-cores. This suggestion 
is based on from the fact that integer programming solvers such as CPLEX become highly competitive for graphs of moderate size, that is, 
between 200 and 2000, and high density, that is, higher than 90\% (e.g.,~see Table~1 in Ref.~\cite{wMaxCliqueExact}).
Since \mbox{$k$-core} decomposition tends to generate relatively high-density and small-sized subgraphs, the combination of the two 
we consider to be a promising avenue for future study.

\section*{Acknowledgement}
The authors would like to thank Marko Bucyk for editing
the manuscript, and Michael Friedlander, Maliheh Aramon, Sourav Mukherjee, 
Natalie Mullin, Jaspreet Oberoi, and Brad Woods for useful comments and discussion.
This work was supported by 1QBit.

\Urlmuskip=0mu plus 2mu\relax
\bibliographystyle{IEEEtran}
\bibliography{bib}

\end{document}